\documentclass[11pt, onecolumn]{IEEEtran}
\usepackage{amsmath}
\usepackage{amsfonts}
\usepackage{amssymb}

\usepackage{psfrag}
\usepackage{enumerate}
\usepackage{url}

\usepackage{graphicx}
\usepackage{color}

\newtheorem{theorem}{Theorem}
\newtheorem{lemma}{Lemma}

\newtheorem{definition}{Definition}

\newcommand{\comment}[1]{}

\def\ISIT{0}

\def\bX{{\bf X}}
\def\bY{{\bf Y}}
\def\bZ{{\bf Z}}

\def\bx{{\bf x}}
\def\by{{\bf y}}

\def\bW{{\bf W}}
\def\bhW{{\bf \hat{W}}}

\def\CR{\dot{R}}
\def\CC{\dot{C}}
\def\CCP{\dot{C}_{\bot}}

\def\by{{\bf y}}

\def\bc{{\bf c}}
\def\bw{{\bf w}}

\def\cA{\mbox{$\cal{A}$}}

\def\cC{\mbox{$\cal{C}$}}

\def\cE{\mbox{$\cal{E }$}}

\def\cB{\mbox{$\cal{B}$}}

\def\cX{\mbox{$\cal{X}$}}
\def\cY{\mbox{$\cal{Y}$}}

\def\cW{\mbox{$\cal{W}$}}

\def\cI{\mbox{$\cal{I}$}}

\def\cS{\mbox{$\cal{S}$}}

\def\cW{\mbox{$\cal{W}$}}
\def\cE{\mbox{$\cal{E}$}}

\def\cN{\mbox{$\cal{N}$}}

\def\d{\mathrm{d}}

\begin{document}
 \IEEEoverridecommandlockouts
 \title{Capacity per Unit-Energy of\\ Gaussian Many-Access Channels}
 \author{
 \IEEEauthorblockN{Jithin~Ravi and Tobias~Koch}\\
 \IEEEauthorblockA{Universidad Carlos III de Madrid, Legan\'es, Spain and Gregorio Mara\~n\'on Health Research Institute, Madrid, Spain\\
       Email: {\tt \{rjithin,koch\}@tsc.uc3m.es}}
 \thanks{J.~Ravi and T.~Koch have received funding from the European Research Council (ERC) under the European Union's Horizon 2020 research and innovation programme (Grant No.~714161). T.~Koch has further received funding from the Spanish Ministerio de Econom\'ia y Competitividad under Grants RYC-2014-16332 and TEC2016-78434-C3-3-R (AEI/FEDER, EU).}
 }

\maketitle
  \begin{abstract}
 We consider a Gaussian multiple-access channel where the number of transmitters grows with the blocklength $n$. For this setup, the maximum number of bits that can be transmitted reliably per unit-energy is analyzed. 
 We show that if the number of users is of an order strictly above $n/\log n$, then the users cannot achieve any positive rate per unit-energy.  In contrast, if the number of users is of order strictly below $n/\log n$, then each user can achieve the single-user capacity per unit-energy $(\log e)/N_0$ (where $N_0/ 2$ is the noise power) by using an orthogonal access scheme such as time division multiple access. We further demonstrate that orthogonal codebooks, which achieve the capacity per unit-energy when the number of users is bounded, can be strictly suboptimal.
 \end{abstract}

 \section{Introduction}
 The capacity per unit-energy $\CC$ is defined as the largest number of bits per unit-energy that can be transmitted reliably over a channel. Verd\'u \cite{Verdu90} showed that $\CC$ can be obtained from the capacity-cost function $C(P)$, defined as the largest number of bits per channel use that can be transmitted reliably with average power per symbol not exceeding $P$, as $\CC=\sup_{P>0} C(P)/P$. For the Gaussian channel with noise power $N_0/2$, this is equal to $\frac{\log e}{N_0}$. Verd\'u further showed that the capacity per unit-energy can be achieved by a codebook that is orthogonal in the sense that the nonzero components of different codewords do not overlap. In general, we shall say that a codebook is orthogonal if the inner product between different codewords is zero. The two-user Gaussian multiple access channel (MAC) was also studied in~\cite{Verdu90}, and it was demonstrated that both users can achieve the single-user capacity per unit-energy by timesharing the channel between the users, i.e., while one user transmits the other user remains silent. This is an orthogonal access scheme in the sense that the nonzero components of codewords of different users do not overlap. In general, we shall say that an access scheme is orthogonal if the inner product between codewords of different users is zero.\footnote{Note, however, that in an orthogonal access scheme the codebooks are not required to be orthogonal. That is, codewords of different codebooks are orthogonal to each other, but codewords of the same codebook need not be.} To summarize, in a two-user Gaussian MAC both users can achieve the rate per unit-energy $\frac{\log e}{N_0} $ by combining an orthogonal access scheme with orthogonal codebooks. This result can be directly generalized to any finite number of users.

The picture changes when the number of users grows without bound with the blocklength $n$. This scenario was investigated recently by Chen \emph{et al.} \cite{ChenCG17}, who referred to such a channel model as a many-access channel (MnAC). Specifically, the MnAC was introduced to model systems consisting of a single receiver and many transmitters, the number of which is comparable to or even larger than the blocklength. This situation could, \emph{e.g.}, occur in a machine-to-machine communication system with many thousands of devices in a given cell. In \cite{ChenCG17}, Chen \emph{et al.} considered a Gaussian MnAC with $k_n$ users and determined the number of messages $M_n$ each user can transmit reliably with a codebook of average power not exceeding $P$. In particular, they showed that the largest sequence $\{M_n\}$ such that the error probability vanishes as $n$ tends to infinity satisfies $\log M_n = \frac{n}{2k_n} \log(1+k_n P) + o(n\log(1+k_n P)/k_n)$. This implies that the per-user rate $(\log M_n)/n$ vanishes as $n\to\infty$ unless $k_n$ is bounded in $n$.

In this paper, we study the capacity per unit-energy of the Gaussian MnAC. We show that, in contrast to the per-user rate, the per-user rate per unit-energy can converge to a positive value as $n\to\infty$ even if $k_n$ grows without bound. Specifically, we demonstrate that, if the order of growth of $k_n$ is strictly below $n/ \log n$, then each user can achieve the capacity per unit-energy $\frac{\log e}{N_0}$ by an orthogonal access scheme. Conversely, if the order of growth of $k_n$ is strictly above $n/ \log n$, then the capacity per unit-energy is zero. Thus, there is a sharp transition between orders of growth where interference-free communication is feasible and orders of growth where reliable communication at any positive rate per unit-energy is infeasible. We further characterize the largest rate per unit-energy that can be achieved with an orthogonal access scheme and orthogonal codebooks. Our characterization shows that orthogonal codebooks are only optimal if $k_n$ grows more slowly than any positive power of $n$.

The paper is organized as follows. In Section~\ref{Sec_model}, we define the problem and introduce some preliminary notations. In Section~\ref{sec_converse}, we present the converse result when the order of $k_n$ is strictly above $n/ \log n$. In Section~\ref{sec_achievbl}, we present the achievability result when the order of $k_n$ is strictly below $n/\log n$. In Section~\ref{sec_ortho}, we analyze the performance of orthogonal codebooks.

\section{Problem Formulation and Preliminaries}
\label{Sec_model}
\subsection{Model and Definitions}
\label{Sec_Def}
Suppose there are $k$ users that wish to transmit their messages $W_i, i=1,\ldots,k$, which are assumed to be independent and uniformly distributed on $\{1,\ldots,M_n^{(i)}\}$, to one common receiver; see Fig.~\ref{Fig_many_acc}. 
To achieve this, they send a codeword of $n$ symbols over the channel. We refer to $n$ as the blocklength. 
We consider a many-access scenario where the number of users $k$ grows with $n$, hence, we denote it as $k_n$.

We further consider a Gaussian channel model where, for $k_n$ users and blocklength $n$, the received vector $\bY$ is given by
\begin{align*}
\bY & = \sum_{i=1}^{k_n} \bX_i(W_i) + \bZ. 
\end{align*}
Here $ \bX_i(W_i)$ is the $n$-length transmitted codeword from user $i$ for message $W_i$ and $\bZ$ is 
a vector of $n$ i.i.d. Gaussian components $Z_j \sim \cN(0, N_0/2)$ independent of $\bX_i$.
We denote the vector of all transmitted codewords by \mbox{$\bX := (\bX_1,\bX_2,\ldots,\bX_{k_n})$}. 
\begin{figure}[]
	\centering
	\includegraphics[scale=0.75]{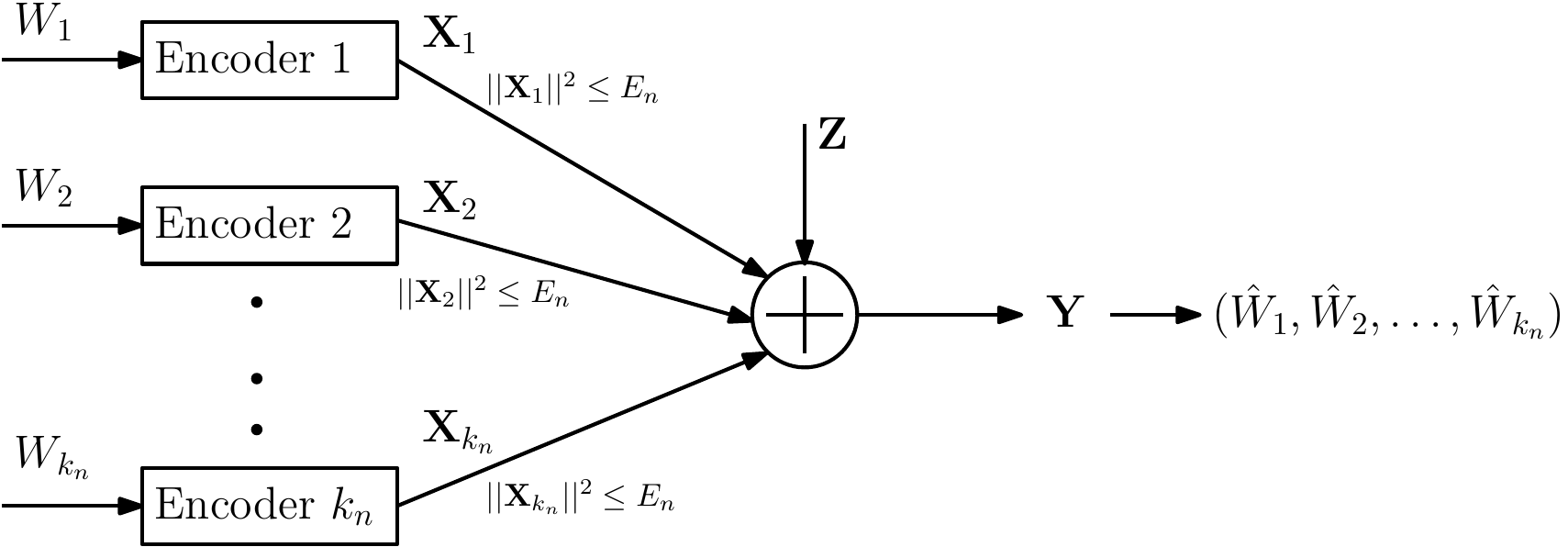}
	\caption{Many-access channel with $k_n$ users at blocklength $n$ }
	\label{Fig_many_acc}
\end{figure}
\begin{definition}
	\label{Def_nMCode}
	For $0 \leq \epsilon < 1$, an  $\bigl(n,\bigl\{M_n^{(\cdot)}\bigr\},\bigl\{E_n^{(\cdot)}\bigr\}, \epsilon\bigr)$-code for the Gaussian many-access channel consists of:
	\begin{enumerate}
		\item Encoding functions $f_i: \{ 1,\ldots,M_n^{(i)}\} \rightarrow \cX^n, i =1,\ldots, k_n$ which map user $i$'s message to the codeword $\bX_i(W_i)$, satisfying the energy constraint
		\begin{align}
		\label{Eq_energy_consrnt}
		\sum_{j=1}^{n} x_{ij}^2(w_i) \leq E_n^{(i)}, 
		\end{align}
		where $x_{ij}$ is the $j$th symbol of the transmitted codeword.
		\item Decoding function $g: \cY^n \rightarrow\{M_n^{(\cdot)}\}$ which maps the received vector $\bY$ to the messages of all users and whose average probability of error satisfies
	\end{enumerate} 
		\begin{align*}
		P_e^{(n)} := P\{ g(\bY) \neq (W_1,\ldots,W_{k_n}) \} \leq \epsilon.
		\end{align*}
\end{definition}
We shall say that
the $\bigl(n,\bigl\{M_n^{(\cdot)}\bigr\},\bigl\{E_n^{(\cdot)}\bigr\}, \epsilon\bigr)$-code is symmetric if $M_n^{(i)} = M_n$ and $E_n^{(i)} = E_n$ for all $i=1, \ldots, k_n$. For compactness, we denote a symmetric code by $(n, M_n, E_n, \epsilon)$. In this paper, we restrict ourselves to symmetric codes.

\begin{definition}
	\label{Def_Sym_Rate_Cost}
	For a symmetric code, the rate per unit-energy $\CR$  is said to be $\epsilon$-achievable if 
	for every $\alpha > 0$, there exists an $n_0$ such that if $n \geq n_0$, then an $(n,M_n,E_n, \epsilon)$-code can be found whose rate per unit-energy satisfies $\frac{\log M_n}{ E_n} > \CR - \alpha$. Furthermore, $\CR$ is said to be achievable if it is $\epsilon$-achievable for all $0 < \epsilon < 1$. The capacity per unit-energy $\CC$ is the supremum of all achievable rates per unit-energy.
\end{definition}

\subsection{Order Notations}
Let $\{a_n\}$ and $\{b_n\}$ be two sequences of nonnegative real numbers.
We write $a_n = O(b_n)$  if there exists an $n_0$ and a positive real number $S$ such that for all $n \geq n_0$, $a_n \leq S b_n$. We write $a_n = o(b_n)$ if $ \lim\limits_{n\rightarrow \infty} \frac{a_n}{b_n} = 0$, and $a_n = \Omega(b_n)$ if $\liminf\limits_{n \rightarrow \infty} \frac{a_n}{b_n} >0$. 
Similarly, $a_n = \Theta (b_n)$ indicates that there exist $ 0 < l_1<l_2$ and $n_0$ such that 
$l_1 b_n \leq a_n \leq l_2 b_n$
for all $n \geq n_0$.  Finally, we write
 $a_n = \omega (b_n)$ if $\lim\limits_{n\rightarrow \infty}  \frac{a_n}{b_n} = \infty$.


 \section{Infeasible Order of Growth}
\label{sec_converse}

We shall refer to orders of $k_n$ for which no positive rate per unit-energy is achievable as infeasible orders of growth. 
In the next theorem, we show that any order of growth which is strictly above $n/\log n$ is infeasible.
\begin{theorem}
	\label{Thm_Infeasble}
	 If $k_n =\omega(n / \log n)$, then $\CC =0$. In words, if the order of $k_n$ is strictly above $n/\log n$, then no coding scheme achieves a positive rate per unit-energy.
\end{theorem}
\begin{proof}
	Let $\bW$ and $\bhW$ denote the vectors $(W_1,\ldots, W_{k_n})$ and $(\hat{W_1},\ldots,\hat{W}_{k_n})$, respectively. Then 
	\begin{align}
	k_n \log M_n& = H(\bW) \nonumber\\
	& = H(\bW|\bhW)+I(\bW;\bhW)\nonumber\\
	& \leq 1+P_e^{(n)}k_n \log M_n + I(\bX;\bY),  \nonumber 
	\end{align}
	by Fano's inequality and the data processing inequality. By following~\cite[Section~9.2]{CoverJ06},
	it can be shown that for the Gaussian channel $I(\bX;\bY) \leq  \frac{n}{2} \log \left(1+\frac{2 k_nE_n}{nN_0}\right)$. Consequently,
	\begin{equation}
	 \frac{\log M_n }{E_n} \leq \frac{1}{k_nE_n}+ \frac{ P_e^{(n)} \log M_n}{E_n}  + \frac{n}{2 k_nE_n} \log \left(\!1+\frac{ 2k_nE_n}{nN_0}\!\right)\!. \notag
	\end{equation}
	This implies that the rate per unit-energy $\CR=(\log M_n)/E_n$ is upper-bounded by 
	\begin{align}
	\CR\leq \frac{ \frac{1}{k_nE_n} + \frac{n}{2 k_nE_n}\log(1+\frac{ 2k_nE_n}{nN_0})}{1 -P_e^{(n)}}.\label{Eq_R_avg}
	\end{align}
	
	We next show by contradiction that if $k_n =\omega(n / \log n)$, then $P_e^{(n)} \to 0$ as $n \to \infty$ only if $\CC=0$. Thus, assume that $k_n =\omega(n / \log n)$ and that there exists a code with rate per unit-energy $\CR >0$ such that $P_e^{(n)} \to 0$ as $n \to \infty$.
	\if \ISIT 1
	To prove that there is a contradiction we need the following lemma whose proof can be found in the extended version of the paper~\cite{RaviK19}.
	\else 
	To prove that there is a contradiction we need the following lemma. 
	\fi
	\begin{lemma}
		\label{Lem_energy_infty}
		If $M_n \geq 2$, then  $P_{e}^{(n)}  \to 0$ only if $E_n \to \infty$.
	\end{lemma}
\begin{proof}See Appendix~\ref{Sec_energy_infty}.
\end{proof}

By the assumption $\CR > 0$, we have that $M_n \geq 2$.
Since we further assumed that 
 $P_{e}^{(n)}  \to 0$, Lemma~\ref{Lem_energy_infty} implies that $E_n \to \infty$. 
 Together with \eqref{Eq_R_avg},  this in turn implies that  $\CR > 0$ is only possible if
 $k_nE_n/n$ is bounded in $n$. Thus,
	\begin{align}
	E_n = O(n/k_n). \label{Eq_energy_bnd}
	\end{align}
The next lemma presents another condition on the order of $E_n$ which contradicts \eqref{Eq_energy_bnd}.
	\begin{lemma}
		\label{Lem_convrs_err_prob}
		 If $\CR > 0$, then $P_{e}^{(n)} \to 0$ only if $E_n = \Omega(\log k_n)$.
	\end{lemma}
\if \ISIT 1
\begin{proof}
See appendix.
\end{proof}
\else
\begin{proof}
	See Appendix~\ref{Sec_Enrgy_log}.
\end{proof}
\fi

We finish the proof by showing that, if $k_n =\omega(n / \log n)$, then there exists no sequence $\{E_n\}$ of order $\Omega(\log k_n)$ that satisfies \eqref{Eq_energy_bnd}. Indeed, $E_n = \Omega(\log k_n)$ and $k_n =\omega(n / \log n)$ imply that
\begin{align}
E_n =\Omega(\log n), \label{Eq_energy_bnd1}
\end{align}
because the order of $E_n$ is lower-bounded by the order of  $\log n - \log \log n$, and  $\log n - \log \log n = \Theta(\log n)$. Furthermore, $E_n = O(n/ k_n)$ and $k_n =\omega(n / \log n)$ imply that 
\begin{align}
E_n & = o(\log n). \label{Eq_energy_bnd2}
\end{align}
Since no sequence $\{E_n\}$ can simultaneously satisfy \eqref{Eq_energy_bnd1} and \eqref{Eq_energy_bnd2}, it follows that, if  $k_n = \omega(n/\log n)$, then no positive rate per unit-energy is achievable.
\end{proof}

 \section{Feasible Order of Growth}
\label{sec_achievbl}

In this section, we show that if the order of 
the growth of $k_n$ is strictly below $n/\log n$, then each user can achieve the single-user capacity per unit-energy $\frac{\log e}{N_0}$. 
Hence, in this case, the users can communicate as if free of interference.
The achievability uses  an orthogonal access scheme, where only one user transmits, all other users remain silent. For further reference, the probability of correct decoding of any orthogonal access scheme is given by 
\begin{align*}
P_c^{(n)} = \prod_{i=1}^{k_n}\left(1-P(\cE_i)\right),
\end{align*}
 where $P(\cE_i)$ denotes the probability of error in decoding  user $i$'s message. 
 In addition, if each user follows the same coding scheme, then the probability of correct decoding is given by 
\begin{align}
P_c^{(n)} & = \left(1-P(\cE_1)\right)^{k_n}. \label{Eq_ortho_prob_err}
\end{align}
We have the following theorem. 
\begin{theorem}
If $k_n = o(n/\log n)$, then any rate per unit-energy satisfying $\CR < \frac{\log e}{N_0}$ is achievable.
Hence, $\CC = \frac{\log e}{N_0}$.
\end{theorem}
\begin{proof}
    For a Gaussian point-to-point channel with power constraint $P$, 
    there exists an encoding and decoding scheme whose
    average probability of error is upper-bounded by
        \begin{align}
    P(\cE) & \leq  M_n^{ \rho}  \exp[-nE_0(\rho, P)],  \; \mbox{ for every } 0< \rho \leq 1, \label{Eq_upp_prob_AWGN} 
    \end{align}
    where 
    \begin{align}
    E_0(\rho, P) & := \frac{\rho}{2} \ln \left(1+\frac{2P}{(1+\rho)N_0}\right). \notag
    \end{align} 
    This bound is due to Gallager and can be found in~\cite[Section~7.4]{Gallager68}.

	Now let us consider an orthogonal access scheme in which each user gets $n/k_n$ channel uses and we timeshare between users. Each user follows
	 the coding scheme which achieves~\eqref{Eq_upp_prob_AWGN} with  power constraint $P_n = \frac{E_n}{n/k_n}$. Note that this satisfies the energy constraint~\eqref{Eq_energy_consrnt}. 
	Then by substituting $n$ with $n/k_n$ and  $P$ with $P_n = \frac{E_n}{n/k_n}$ in \eqref{Eq_upp_prob_AWGN}, we get the following bound on $P(\cE_{1})$ as a function of the rate per unit-energy $\CR = \frac{\log M_n}{E_n}$:
	\begin{align}
	 P(\cE_{1}) & \leq  M_n^{ \rho}  \exp\left[-\frac{ n}{k_n}E_0(\rho, P_n)\right] \nonumber \\
	& = \exp\left[   \rho  \ln M_n -   \frac{ n}{k_n} \frac{\rho}{2} \ln \left(1+\frac{ 2E_nk_n/n}{(1+\rho)N_0}\right) \right] \nonumber \\
	& = \exp\left[ -E_n \rho   \left(  \frac{\ln (1+\frac{ 2E_nk_n/n}{(1+\rho)N_0})}{2E_nk_n/n } -\frac{\CR}{ \log e} \right)\right]. \label{Eq_err_uppr}
	\end{align}
Combining \eqref{Eq_err_uppr} with \eqref{Eq_ortho_prob_err}, we obtain that the probability of correct decoding can be lower-bounded as 
	\begin{align}
&  1 -P_e^{(n)}   \geq \Biggl(1 -   \exp\Biggl[ -E_n \rho   \Biggl(  \frac{\ln (1+\frac{ 2E_nk_n/n}{(1+\rho)N_0})}{2E_nk_n/n } -\frac{\CR}{ \log e} \Biggr)\Biggr]\Biggr)^{k_n}. \label{Eq_ortho_prob_lower}
\end{align}
We next choose $E_n = c_n \ln n$ with $c_n := \ln\bigl(\frac{n}{k_n\ln n}\bigr)$. Since, by assumption, $k_n = o(n / \log n)$, this implies that $\frac{k_nE_n}{n} \to 0$ as $n \to \infty$. Consequently, the first term in the inner most bracket in \eqref{Eq_ortho_prob_lower} tends to $1/((1+\rho)N_0)$ as $n \to \infty$. It follows that for $\CR < \frac{\log e}{N_0}$, there exists a sufficiently large $n_0$, a $\rho > 0$, and a $\delta>0$ such that, for $n\geq n_0$, the RHS of \eqref{Eq_ortho_prob_lower} is lower-bounded by $\left(1-\exp[-E_n \rho \delta]\right)^{k_n}$. Since $c_n\delta \rho \to \infty$ as $n\to\infty$, we have
	\begin{align}
	\left(1-\exp[-E_n \rho \delta]\right)^{k_n}& \geq \left(1-\frac{1}{n^{2}}\right)^{k_n} \notag \\
	& \geq	\left(1-\frac{1}{n^{2}}\right)^{\frac{n}{\log n}} \notag \\
	& =  \left[\left(1 - \frac{1}{n^{2}}\right)^{n^{2}}\right]^{\frac{1}{n\log n}}, \label{Eq_prob_corrct}
	\end{align}
	for sufficiently large $n\geq n_0$ such that \mbox{$c_n\delta \rho\geq2$}  and  \mbox{$k_n \leq \frac{n}{\log n}$}. Noting that $(1 - \frac{1}{n^{2}})^{n^{2}} \to 1/e$ and $\frac{1}{n\log n} \to 0$ as $n\to\infty$, we obtain that the
	RHS of \eqref{Eq_prob_corrct} goes to one as $n \to \infty$. This implies that, if $k_n = o(n/\log n)$, then any rate per unit-energy  $\CR < \frac{\log e}{ N_0} $ is achievable.
\end{proof}

 \section{Performance of Orthogonal Codebooks}
\label{sec_ortho}

As mentioned in the introduction, when the number of users is bounded, the capacity per unit-energy $\CC = \frac{\log e}{N_0}$ can be achieved with orthogonal codebooks. In the following theorem, we characterize the largest rate per unit-energy achievable with orthogonal codebooks, denoted by $\CCP$, when the number of users grows with the blocklength.


\begin{theorem} 
	\label{Thm_ortho_code}
	Suppose the users apply an orthogonal access scheme and each user uses orthogonal codebooks. Then:
	\begin{enumerate}[1)]
		\item 	 If $k_n = o(n^{c})$ for every $c>0$, then  $\CCP = \frac{\log e}{N_0}$. \label{Thm_ortho_part1}
		\item If $k_n=\Theta\left({n^c}\right)$, then 
		\begin{equation*}
		\CCP = \begin{cases} \frac{\log e}{N_0}  \frac{1}{\left(1+\sqrt{\frac{c}{1-c}}\right)^2}, \quad & \textnormal{if $0<c\leq 1/2$} \\ \frac{\log e}{2 N_0} (1-c), \quad & \textnormal{if $1/2<c<1$}.\end{cases}
		\end{equation*}
		\label{Thm_ortho_part2}
	\end{enumerate}
\end{theorem}

\begin{proof}
To prove Theorem~\ref{Thm_ortho_code}, we shall first present in the following lemma bounds on the probability of error achievable over a Gaussian point-to-point channel with an orthogonal codebook.
\if \ISIT 1
	The bounds are obtained using similar methods given in~\cite[Section~2.5]{ViterbiO79} and the proof can be found in the extended version of this paper~\cite{RaviK19}.
	\else 

	\fi
	\begin{lemma}
		\label{Lem_ortho_code}
		For an orthogonal codebook with $M$ codewords  of energy less than or equal to $E$, the probability of error satisfies the following bounds:
		\begin{enumerate}
			\item For  $0 < \CR \leq \frac{1}{4} \frac{\log e}{N_0}$,
			\begin{align}
			&  \exp\left[- \frac{\ln M}{\CR}\left(\frac{\log e}{2 N_0} - \CR + o(1) \right) \right] \leq P_e  \leq  \exp\left[- \frac{\ln M}{\CR}\left(\frac{\log e}{2 N_0} - \CR \right) \right]. \label{Eq_orth_sinlg_uppr1}
			\end{align}
			\item For $\frac{1}{4} \frac{\log e}{N_0}  \leq \CR \leq \frac{\log e}{N_0}$,
		    \begin{align}
			&  \exp\left[- \frac{\ln M}{\CR}\left(\left(\sqrt{\frac{\log e}{N_0}}- \sqrt{\CR }\right)^2 + o(1) \right) \right] \leq P_e  \leq  \exp\left\{- \frac{\ln M}{\CR}  \left(\sqrt{\frac{\log e}{N_0}} - \sqrt{ \CR}\right)^2  \right\}. \label{Eq_orth_sinlg_uppr2}
			\end{align}
		\end{enumerate}
	In~\eqref{Eq_orth_sinlg_uppr1} and~\eqref{Eq_orth_sinlg_uppr2}, $o(1) \to 0$ as $E \to \infty$. 
	\end{lemma}
\if \ISIT 1

\else
\begin{proof}
See Appendix~\ref{Sec_AWGN_ortho_code}.
	\end{proof}

\fi

Next, we define
\begin{align}
a := \left\{
\begin{array}{cl}
\frac{\left(\frac{\log e}{2 N_0} - \CR  \right) }{\CR}, & \quad \mbox{if  } 0 < \CR \leq \frac{1}{4} \frac{\log e}{N_0} \\
\frac{ \left(\sqrt{\frac{\log e}{N_0}} - \sqrt{ \CR}\right)^2  }{\CR},  & \quad \mbox{if } \frac{1}{4} \frac{\log e}{N_0}  \leq \CR \leq \frac{\log e}{N_0}
\end{array}\right. \label{Eq_def_a}
\end{align}
and let $a_E:= a + o(1)$.
Then the bounds in Lemma~\ref{Lem_ortho_code} can be written as 
\begin{align}
1/M^{a_E} & \leq P_e \leq 1/M^{a}. \label{Eq_prob_err_singl}
\end{align}

	Now let us consider the case where the users apply an orthogonal access scheme and each user uses an orthogonal codebook. 
	For an orthogonal access scheme with orthogonal codebooks, the collection of codewords from all users is orthogonal, hence there are at most $n$ codewords of length $n$. Since with a symmetric code each user transmits the same number of messages, it follows that each user transmits $M_n=n/k_n$  messages with 
	codewords of energy less than or equal to $E_n$.
	 In this case, we obtain from \eqref{Eq_ortho_prob_err} and \eqref{Eq_prob_err_singl} that 
	\begin{align*}
	\left(1-\left(\frac{k_n}{n}\right)^{a}\right)^{k_n} \leq \left(1- P(\cE_1)\right)^{k_n} \leq \left(1-\left(\frac{k_n}{n}\right)^{a_{E_n}}\right)^{k_n},
	\end{align*}
	which, denoting $a_n := a_{E_n}$, can be written as 
	\begin{align}
	& \left[\left(1-\left(\frac{k_n}{n}\right)^a\right)^{(\frac{n}{k_n})^a}\right]^{\frac{k_n^{1+a}}{n^a}}\leq \left(1- P(\cE_1)\right)^{k_n}
	  \leq \left[\left(1-\left(\frac{k_n}{n}\right)^{a_n}\right)^{(\frac{n}{k_n})^{a_n}}\right]^ {\frac{k_n^{1+a_n}}{n^{a_n}}}. \label{Eq_ortho_code_upp_low}
	\end{align}
	Since  Theorem~\ref{Thm_ortho_code} 
	only concerns a sublinear number of users, we have
	\begin{align*}
	\lim\limits_{n\to \infty} \left(1-\left(\frac{k_n}{n}\right)^a\right)^{(\frac{n}{k_n})^a} 
	& =  \frac{1}{e}.
	\end{align*}
	Furthermore, if $P_e^{(n)} \to 0$ then by Lemma~\ref{Lem_energy_infty} $E_n \to \infty$ as $n \to \infty$, in which  case $a_n$ converges to the finite value $a$ as $n \to \infty$, and we obtain 
	\begin{align*}
	 \lim\limits_{n \to \infty}  \left(1-\left(\frac{k_n}{n}\right)^{a_n}\right)^{(\frac{n}{k_n})^{a_n}} & =  \frac{1}{e}.
	\end{align*}
	 So \eqref{Eq_ortho_code_upp_low} implies that 
	 $P_e^{(n)} \to 0$ as $n \to \infty$ if
	 \begin{align}
	\lim\limits_{n \to \infty}	{\frac{k_n^{1+a}}{n^a}} = 0,\label{Eq_ordr_lowr}
	 \end{align}
	 and only if 
	 \begin{align}
	\lim\limits_{n \to \infty}	 {\frac{k_n^{1+a_n}}{n^{a_n}}} = 0.  \label{Eq_ordr_uppr}
	 \end{align}
We next use these observation to prove Parts~\ref{Thm_ortho_part1}) and \ref{Thm_ortho_part2}) of Theorem~\ref{Thm_ortho_code}. We begin with Part~\ref{Thm_ortho_part1}). Let $\CR < \frac{\log e}{N_0}$. Thus, we have $a>0$ which implies that we can find a constant $\eta < a/(1+a)$ such that $n^{\eta (1+a)}/n^a \to 0$ as $n \to \infty$. Since by assumption $k_n =o(n^c)$ for every $c> 0$, it follows that there exists an $n_0$ such that, for all $n\geq n_0$, we have $k_n \leq n ^{\eta(1+a)}$. This implies that  \eqref{Eq_ordr_lowr} is satisfied, from which Part~\ref{Thm_ortho_part1}) follows.
	


We next prove Part~\ref{Thm_ortho_part2}) of Theorem~\ref{Thm_ortho_code}.
Indeed,	if $k_n=\Theta\left({n^c}\right)$,  $0<c<1$, then there exist $0<l_1<l_2$ and $n_0$ such that, for all $n\geq n_0$, we have $(l_1n)^c\leq k_n \leq (l_2n)^c$. Consequently,
	 \begin{align}
	   {\frac{(l_1 n)^{c(1+a_n)}}{n^{a_n}}} \leq    {\frac{k_n^{1+a_n}}{n^{a_n}}}
	  \leq {\frac{(l_2 n)^{c(1+a_n)}}{n^{a_n}}}. \label{Eq_ortho_err_uppr}
	 \end{align}
	 If $P_e^{(n)}  \to 0$, then from~\eqref{Eq_ordr_uppr} we have ${\frac{k_n^{1+a_n}}{n^{a_n}}}\to 0$. Thus, \eqref{Eq_ortho_err_uppr} implies that  $c(1+a_n) - a_n$ converges to a negative value.
	Since $c(1+a_n) - a_n$ tends to $c(1+a) - a$ as $n \to \infty$, it follows that  $P_e^{(n)}  \to 0$ only if  $c(1+a) - a < 0$, which is the same as $a > c/(1-c)$. 
	Using similar arguments, it follows from \eqref{Eq_ordr_lowr}  that if  $a > c/(1-c)$, then $P_e^{(n)}  \to 0$. Hence, $P_e^{(n)}  \to 0$ if and only if $a > c/(1-c)$.
	It can be observed from~\eqref{Eq_def_a} that $a$  is a monotonically decreasing function of $\CR$. So for $k_n=\Theta\left({n^c}\right), 0<c<1$, the capacity per unit-energy $\CCP$  is given by
\begin{align}
\CCP = \sup \{\CR\geq 0 : a(\CR) > c/(1-c)\}, \notag
\end{align}
where we write $a(\CR)$ to make it clear that $a$ as defined in \eqref{Eq_def_a} is a function of $\CR$.
 This supremum can be computed as 
\begin{equation*}
\CCP = \begin{cases} \frac{\log e}{N_0}  \left(\frac{1}{1+\sqrt{\frac{c}{1-c}}}\right)^2, & \quad \mbox{if  } 0 < c \leq  1/2\\
\frac{\log e}{2 N_0} (1-c), & \quad \mbox{if } 1/2<c < 1,
\end{cases}
\end{equation*}
which proves Part~\ref{Thm_ortho_part2}) of Theorem~\ref{Thm_ortho_code}.
\end{proof}


\if \ISIT 1

\appendix[Proof of Lemma~\ref{Lem_convrs_err_prob}]
\label{sec_lem_proof}
Let $\cW$ denote the set of  $M_n^{k_n}$ messages  of all users at  blocklength $n$.
To prove the lemma, we first show that 	
\begin{align*}
\frac{1}{M_n^{k_n}} \sum_{\bw \in \cW} P_e(\bw) \geq  1  -   \frac{ 8 E_n+\log 2}{\log \left( k_n(M_n -1)\right)},
\end{align*}
where $P_e(\bw)$ denotes the probability of error in decoding the set of messages $\bw=(w_1,\ldots,w_{k_n})$. To this end, we show that there exists a partition $\cS_d$, $d=1, \ldots,D$ of $\cW$ such that for every $d$ we have 
\begin{align}
	\frac{1}{|\cS_d|} \sum_{\bw \in \cS_d} P_e(\bw) \geq  1  -   \frac{ 8 E_n+\log 2}{\log \left( k_n(M_n -1)\right)}, \label{Eq_avg_prob_lowr}
\end{align}
where $|\cdot|$ denotes the cardinality of the respective set.
This implies that the RHS of \eqref{Eq_avg_prob_lowr} is also a lower bound on
\begin{align*}
\frac{1}{M_n^{k_n}} \sum_{\bw \in \cW} P_e(\bw) = \frac{1}{M_n^{k_n}} \sum_{d =1}^{D} \sum_{\bw \in \cS_d} P_e(\bw).  
\end{align*}

To describe the partition, we use the following representation for $\bw\in \cW$: 
Each $\bw \in \cW$ is denoted using a $k_n$-length vector such that the $i^{\mathrm{th}}$ position of the vector is set to $j$ if user $i$ has message $j$, where $1\leq j \leq M_n$.
The Hamming distance $d_H$ between two messages $\bw=(w_1,\ldots,w_{k_n})$ and $\bw'=(w'_1,\ldots,w'_{k_n})$ is defined as the number of positions at which $\bw$ differs from $\bw'$, i.e.,
$d_H(\bw,\bw') := \left|\{i: w_i\neq w'_i \}\right|$.

We next show that one can find a partition $\cS_d$, $d=1, \ldots,D$  such that for each set $\cS_d$, there exists a vector from which all vectors in $\cS_d$ are at Hamming distance at most two.
Let $C$ be a code in $\cW$ with minimum Hamming distance $3$, such that for any $\bw\in \cW$ there exists at least one codeword in $C$ which is at most at a distance 2 from it. 
Such a code exists because if for some $\bw\in \cW$ all codewords were at a distance 3 or more, then we could add $\bw$ to code $C$ without affecting its minimum distance. Let $\bc(1),\ldots, \bc(|C|)$ denote the codewords of code $C$. Next we partition the set $\cW$ into $D=|C|$ sets, $\cS_d$, $d=1,\ldots, D$, as follows:

For a given $d=1, \ldots,D$, we assign $\bc(d)$ to $\cS_d$ as well as all $\bw \in \cW$ that satisfy $d_H(\bw, \bc(d))=1$. These assignments are unique since the code $C$ has minimum Hamming distance 3.
We next consider  all $\bw\in \cW$ for which there is no codeword $\bc(1), \ldots, \bc(|C|)$ satisfying $d_H(\bw, \bc(d))= 1$ and assign it  to the set with index $d = \min \{i=1,\ldots, D: d_H(\bw, \bc(i)) =2 \}$. Like this, we obtain a partition of $\cW$, and since 
for every codeword 
there are $k_n(M_n -1)$ sequences at Hamming distance one from it, this partition satisfies $|\cS_d| \geq 1+k_n(M_n -1), d=1,\ldots, D$.




 We next derive the lower bound \eqref{Eq_avg_prob_lowr}.	To this end, we use  a stronger form of Fano's inequality known as Birg\'e's inequality~\cite{Yatracos88}.
\begin{lemma}[Birg\'e's inequality]
	\label{Lem_Berge}
	Let $(\cY, \cB)$ be a measurable space with a $\sigma$-field and $P_1,\ldots, P_N$  be probability measures defined on $\cB$. Further let $\cA_i$, $i=1, \ldots,N$  denote $N$ events defined on $\cY$, where $N\geq 2$.
	Then 
	\begin{align*}
	\frac{1}{N} \sum_{i=1}^{N} P_i(\cA_i) \leq \frac{\frac{1}{N^2} \sum_{i,j}^{}  D(P_i\|P_j)+\log 2}{\log (N-1)}.
	\end{align*}
\end{lemma}

To apply Lemma~\ref{Lem_Berge} to the problem at hand, we set $N=|\cS_d|$ and $P_i = P_{Y|\bX}(\cdot|\bx(i))$, where $\bx(i)$ denotes the set of codewords transmitted to convey the set of messages $i \in \cS_d$.
We further let $\cA_i$ denote the subset of $\cY^n$ for which the decoder declares the set of messages $i\in\cS_d$. Then, the probability of
error in decoding messages $i\in\cS_d$ is given by $P_e(i) =1-P_i(\cA_i)$ and $\frac{1}{|\cS_d|} \sum_{i\in \cS_d} P_i(\cA_i)$ denotes the average probability of correctly decoding a message in $\cS_d$.

For two multivariate Gaussian distributions \mbox{${\bf Z}_1 \sim \cN(\boldsymbol {\mu_1 }, \frac{N_0}{2}\cI)$}
 and ${\bf Z}_1 \sim \cN(\boldsymbol {\mu_2}, \frac{N_0}{2}\cI)$,
 the relative entropy $D({\bf Z}_1\| { \bf Z}_2)$ is given  by $ \frac{ ||\boldsymbol {\mu_1 - \mu_2}||^2}{N_0}$. We next note that $P_{\bw} =  \cN(\overline{\bx}(\bw), \frac{N_0}{2}\cI)$ and $P_{\bw'} = \cN(\overline{\bx}(\bw'), \frac{N_0}{2}\cI)$, where $\overline{\bx}(i)$ denotes the sum of codewords contained in $\bx(i)$.
 Since the energy of a codeword for any user is upper bounded by $E_n$, and since any two  $\bw, \bw' \in \cS_d$ are at a Hamming distance of at most 4, we get that $||\overline{\bx}(\bw) -\overline{\bx}(\bw')||^2 \leq 64 E_n$.
  Consequently,
\begin{align*}
D(P_{\bw}\|P_{\bw'}) \leq 64 E_n/N_0, \quad \bw, \bw'\in\cS_d. 
\end{align*}
It thus follows from Birg\'e's inequality that
\begin{align}
\frac{1}{|\cS_d|} \sum_{\bw\in \cS_d} P_e(\bw) & \geq 1  -   \frac{ 64 E_n/N_0+\log 2}{\log (|\cS_d|-1)} \notag \\
& \geq 1  -   \frac{ 64 E_n/N_0+\log 2}{\log \left( k_n(M_n -1)\right)}, \label{Eq_Nc_lowr_bnd}
\end{align}
where \eqref{Eq_Nc_lowr_bnd} follows because $ |\cS_d|-1 \geq k_n(M_n -1)$. 
Note that \eqref{Eq_Nc_lowr_bnd}  holds for all $d=1,\ldots, D$, so
\begin{align*}
P_{e}^{(n)} \geq  1  -   \frac{ 64 E_n/N_0+\log 2}{\log \left( k_n(M_n -1)\right)}.
\end{align*}
This shows that $P_{e}^{(n)}$ goes to zero only if 
\begin{align}
E_n & = \Omega\left(\log (k_n(M_n -1))\right)\notag \\
& = \Omega\left(\log M_n + \log k_n\right), \label{Eq_enrgy_ordr}
\end{align}
where the second step follows because, by Lemma~\ref{Lem_energy_infty}, $P_{e}^{(n)} \to 0$ as $n \to \infty$ only if $E_n \to \infty$,
which by the assumption $\CR > 0$ implies that $M_n \to \infty$. Using that $\log M_n = E_n\CR$, \eqref{Eq_enrgy_ordr} can be written as $E_n = \Omega(E_n\CR + \log k_n)$, which is equivalent to $E_n = \Theta(E_n\CR + \log k_n)$. However, this holds only if \mbox{$\log k_n = O(E_n)$}, which is equivalent to $E_n = \Omega(\log k_n)$. This proves Lemma~\ref{Lem_convrs_err_prob}.	

\else
\appendices

\section{Proof of Lemma~\ref{Lem_energy_infty}}
\label{Sec_energy_infty}
The probability of error of the Gaussian MnAC cannot be smaller than that of the Gaussian point-to-point channel. Indeed, suppose a genie informs the receiver about all transmitted codewords except that of user $i$. Then the receiver can subtract the known codewords from the received vector, resulting in a point-to-point Gaussian channel. Since access to additional information does not increase the probability of error, the claim follows.

We next note that, for a Gaussian point-to-point channel, any $(n,M_n, E_n, \epsilon)$-code satisfies~\cite[Theorem~2]{PolyanskiyPV11}
	\begin{align}
	\frac{1}{M_n} \geq Q\left(\sqrt{\frac{2E_n}{N_0}}+Q^{-1}(1-\epsilon)\right), \label{Eq_finite_energy}
	\end{align}
	where $Q$ denotes the tail distribution function of the standard Gaussian distribution, and $Q^{-1}$ denotes its inverse function. Solving~\eqref{Eq_finite_energy}  for $\epsilon$ yields
	\begin{align}
	\epsilon  &\geq 1 - 	Q\left(Q^{-1}\left(\frac{1}{M_n}\right)  - \sqrt{\frac{2E_n}{N_0}}\right). \nonumber
	\end{align}
	It follows that the probability of error tends to zero as $n \to \infty$ only if $Q^{-1}\left(\frac{1}{M_n}\right)  - \sqrt{\frac{2E_n}{N_0}} \rightarrow -\infty$.
	Since $Q^{-1}\left(\frac{1}{M_n}\right) \geq 0$ for $M_n\geq 2 $, this in turn is only the case if
   $E_n \rightarrow \infty$. This proves Lemma~\ref{Lem_energy_infty}.

\section{Proof of Lemma~\ref{Lem_convrs_err_prob}}
\label{Sec_Enrgy_log}
Let $\cW$ denote the set of  $M_n^{k_n}$ messages  of all users at  blocklength $n$.
To prove the lemma, we first show that 	
\begin{align*}
\frac{1}{M_n^{k_n}} \sum_{\bw \in \cW} P_e(\bw) \geq  1  -   \frac{ 64 E_n/N_0+\log 2}{\log \left( k_n(M_n -1)\right)},
\end{align*}
where $P_e(\bw)$ denotes the probability of error in decoding the set of messages $\bw=(w_1,\ldots,w_{k_n})$. To this end, we show that there exists a partition $\cS_d$, $d=1, \ldots,D$ of $\cW$ such that for every $d$ we have 
\begin{align}
\frac{1}{|\cS_d|} \sum_{\bw \in \cS_d} P_e(\bw) \geq  1  -   \frac{ 64 E_n/N_0+\log 2}{\log \left( k_n(M_n -1)\right)}, \label{Eq_avg_prob_lowr}
\end{align}
where we use $|\cdot|$ to denote the cardinality of a set.
This implies that the RHS of \eqref{Eq_avg_prob_lowr} is also a lower bound on
\begin{align*}
\frac{1}{M_n^{k_n}} \sum_{\bw \in \cW} P_e(\bw) = \frac{1}{M_n^{k_n}} \sum_{d =1}^{D} \sum_{\bw \in \cS_d} P_e(\bw).  
\end{align*}

To describe the partition, we use the following representation for $\bw\in \cW$: 
Each $\bw \in \cW$ is denoted using a $k_n$-length vector such that the $i^{\mathrm{th}}$ position of the vector is set to $j$ if user $i$ has message $j$, where $1\leq j \leq M_n$.
The Hamming distance $d_H$ between two messages $\bw=(w_1,\ldots,w_{k_n})$ and $\bw'=(w'_1,\ldots,w'_{k_n})$ is defined as the number of positions at which $\bw$ differs from $\bw'$, i.e.,
$d_H(\bw,\bw') := \left|\{i: w_i\neq w'_i \}\right|$.

We next show that one can find a partition $\cS_d$, $d=1, \ldots,D$  such that for each set $\cS_d$, there exists a vector from which all vectors in $\cS_d$ are at Hamming distance at most two.
Let $\cC$ be a code in $\cW$ with minimum Hamming distance $3$, such that for any $\bw\in \cW$ there exists at least one codeword in $\cC$ which is at most at a distance 2 from it. 
Such a code exists because if for some $\bw\in \cW$ all codewords were at a distance 3 or more, then we could add $\bw$ to $\cC$ without affecting its minimum distance. 
Thus for all $\bw \notin \cC$, there exists at least one $i$ such that $d_H(\bw,\bc(i)) \leq 2$.
Let $\bc(1),\ldots, \bc(|\cC|)$ denote the codewords of code $\cC$. Next we partition the set $\cW$ into $D=|\cC|$ sets, $\cS_d$, $d=1,\ldots, D$, as follows:

For a given $d=1, \ldots,D$, we assign $\bc(d)$ to $\cS_d$ as well as all $\bw \in \cW$ that satisfy $d_H(\bw, \bc(d))=1$. These assignments are unique since the code $\cC$ has minimum Hamming distance 3.
We next consider  all $\bw\in \cW$ for which there is no codeword $\bc(1), \ldots, \bc(|\cC|)$ satisfying $d_H(\bw, \bc(d))= 1$ and assign it  to the set with index $d = \min \{i=1,\ldots, D: d_H(\bw, \bc(i)) =2 \}$. Like this, we obtain a partition of $\cW$, and since 
for every codeword 
there are $k_n(M_n -1)$ sequences at Hamming distance one from it, this partition satisfies $|\cS_d| \geq 1+k_n(M_n -1), d=1,\ldots, D$.

We next derive the lower bound \eqref{Eq_avg_prob_lowr}.	To this end, we use  a stronger form of Fano's inequality known as Birg\'e's inequality.
\begin{lemma}[Birg\'e's inequality]
	\label{Lem_Berge}
	Let $(\cY, \cB)$ be a measurable space with a $\sigma$-field and $P_1,\ldots, P_N$  be probability measures defined on $\cB$. Further let $\cA_i$, $i=1, \ldots,N$  denote $N$ events defined on $\cY$, where $N\geq 2$.
	Then 
	\begin{align*}
	\frac{1}{N} \sum_{i=1}^{N} P_i(\cA_i) \leq \frac{\frac{1}{N^2} \sum_{i,j}^{}  D(P_i\|P_j)+\log 2}{\log (N-1)}.
	\end{align*}
\end{lemma}
\begin{proof}
	See~\cite{Yatracos88} and references therein.
\end{proof}
To apply Lemma~\ref{Lem_Berge} to the problem at hand, we set $N=|\cS_d|$ and $P_i = P_{Y|\bX}(\cdot|\bx(i))$, where $\bx(i)$ denotes the set of codewords transmitted to convey the set of messages $i \in \cS_d$.
We further let $\cA_i$ denote the subset of $\cY^n$ for which the decoder declares the set of messages $i\in\cS_d$. Then, the probability of
error in decoding messages $i\in\cS_d$ is given by $P_e(i) =1-P_i(\cA_i)$, and $\frac{1}{|\cS_d|} \sum_{i\in \cS_d} P_i(\cA_i)$ denotes the average probability of correctly decoding a message in $\cS_d$.

For two multivariate Gaussian distributions \mbox{${\bf Z}_1 \sim \cN(\boldsymbol {\mu_1 }, \frac{N_0}{2}I)$}
and ${\bf Z}_1 \sim \cN(\boldsymbol {\mu_2}, \frac{N_0}{2}I)$ (where $I$ denotes the identity matrix),
the relative entropy $D({\bf Z}_1\| { \bf Z}_2)$ is given  by $ \frac{ ||\boldsymbol {\mu_1 - \mu_2}||^2}{N_0}$. We next note that $P_{\bw} =  \cN(\overline{\bx}(\bw), \frac{N_0}{2}I)$ and $P_{\bw'} = \cN(\overline{\bx}(\bw'), \frac{N_0}{2}I)$, where $\overline{\bx}(i)$ denotes the sum of codewords contained in $\bx(i)$.
 Any two vectors $\bw, \bw' \in \cS_d$ are at a Hamming distance of at most 4. Without loss of generality, let us assume that $w_i = w'_i$ for $i=5, \ldots, k_n$. Then
 \begin{align}
\left\|\sum_{i=1}^{k_n} \bx_i(w_i) -\sum_{i=1}^{k_n} \bx_i(w'_i)\right\|^2 & = \left\|\sum_{i=1}^{4} \bx_i(w_i) - \bx_i(w'_i)\right\|^2 \notag \\
& \leq \left\|\sum_{i=1}^{4} |\bx_i(w_i) - \bx_i(w'_i)|\right\|^2 \notag\\
& \leq (4 \times 2\sqrt{E_n})^2 \notag \\
& = 64 E_n, \notag
 \end{align}
 where the first inequality follows because of the triangle inequality $|a+b| \leq |a|+|b|$, 
and the second inequality follows since the energy of a codeword for any user is upper-bounded by $E_n$ and since $|a-b|\leq |a|+|b|$.
  Consequently,
\begin{align*}
D(P_{\bw}\|P_{\bw'}) \leq 64 E_n/N_0, \quad \bw, \bw'\in\cS_d. 
\end{align*}
We thus obtain from Birg\'e's inequality that
\begin{align}
\frac{1}{|\cS_d|} \sum_{\bw\in \cS_d} P_e(\bw) & \geq 1  -   \frac{ 64 E_n/N_0+\log 2}{\log (|\cS_d|-1)} \notag \\
& \geq 1  -   \frac{ 64 E_n/N_0+\log 2}{\log \left( k_n(M_n -1)\right)}, \label{Eq_Nc_lowr_bnd}
\end{align}
where \eqref{Eq_Nc_lowr_bnd} follows because $ |\cS_d|-1 \geq k_n(M_n -1)$. 
Note that \eqref{Eq_Nc_lowr_bnd}  holds for all $d=1,\ldots, D$, so
\begin{align*}
P_{e}^{(n)} \geq  1  -   \frac{ 64 E_n/N_0+\log 2}{\log \left( k_n(M_n -1)\right)}.
\end{align*}
This shows that $P_{e}^{(n)}$ goes to zero only if 
\begin{align}
E_n & = \Omega\left(\log (k_n(M_n -1))\right)\notag \\
& = \Omega\left(\log M_n + \log k_n\right), \label{Eq_enrgy_ordr}
\end{align}
where the second step follows because, by Lemma~\ref{Lem_energy_infty}, $P_{e}^{(n)} \to 0$ as $n \to \infty$ only if $E_n \to \infty$,
which by the assumption $\CR > 0$ implies that $M_n \to \infty$. Using that $\log M_n = E_n\CR$, \eqref{Eq_enrgy_ordr} can be written as $E_n = \Omega(E_n\CR + \log k_n)$, which is equivalent to $E_n = \Theta(E_n\CR + \log k_n)$. However, this holds only if \mbox{$\log k_n = O(E_n)$}, which is equivalent to $E_n = \Omega(\log k_n)$. This proves Lemma~\ref{Lem_convrs_err_prob}.	

\section{Proof of Lemma~\ref{Lem_ortho_code}}
\label{Sec_AWGN_ortho_code}
The upper bounds on the probability of error in~\eqref{Eq_orth_sinlg_uppr1} and \eqref{Eq_orth_sinlg_uppr2} are proved in Subsection~\ref{Sec_uppr}. The lower bounds are proved in Subsection~\ref{Sec_lowr}.
\subsection{Upper bounds}
\label{Sec_uppr}
An upper bound on the probability of error for $M$ orthogonal codewords of maximum energy $E$ can be found in~\cite[Section~2.5]{ViterbiO79}:
\begin{align}
P_{e} & \leq (M-1)^{\rho} \exp\left[-\frac{E}{N_0} \left (\frac{\rho}{\rho+1}\right)\right]  \notag\\
& \leq  \exp\left[-\frac{E}{N_0} \left(\frac{\rho}{\rho+1}\right)+ \rho \ln M\right], \text{ for all }  0 \leq \rho\leq 1.\label{Eq_enrgy_AWGN}
\end{align}
For the rate per unit-energy $\CR = \frac{ \log M}{E}$, it follows from~\eqref{Eq_enrgy_AWGN}
that
\begin{align}
P_{e}  & \leq  \exp\left[-\frac{E}{N_0} \left(\frac{\rho}{\rho+1}\right)+ \frac{\rho E \CR}{\log e} \right]\notag\\
& = \exp[-E E_0(\rho, \CR)], \label{Eq_achvbl_err_exp}
\end{align}
where
\begin{align}
E_0(\rho, \CR) & := \left(\frac{1}{N_0} \frac{\rho}{\rho+1} -\frac{\rho \CR}{\log e}\right). \label{Eq_err_exp}
\end{align}
When $ 0 < \CR \leq \frac{1}{4} \frac{\log e}{N_0} $, the maximum of $E_0(\rho, \CR)$ over all $0 \leq \rho \leq 1$ is achieved for $\rho=1$.
When $\frac{1}{4} \frac{\log e}{N_0} \leq \CR \leq  \frac{\log e}{N_0}$, the maximum of $ E_0(\rho, \CR) $  is achieved for $\rho = \sqrt{ \frac{\log e}{N_0} \frac{1}{\CR} } -1 \in[0,1]$.   
It follows that
 \begin{align}
\max_{0 \leq \rho \leq 1 } E_0(\rho, \CR) =
 \begin{cases}
\frac{1}{2 N_0} - \frac{\CR}{\log e}, & 0 < \CR \leq \frac{1}{4} \frac{\log e}{N_0}  \\
\left(\sqrt{\frac{1}{N_0}} - \sqrt{ \frac{\CR}{\log e}}\right)^2,   & \frac{1}{4} \frac{\log e}{N_0}  \leq \CR \leq \frac{\log e}{N_0} .
 \end{cases}
  \label{Eq_err_exp_achv}
 \end{align}
 
Since $E = \frac{\log M}{\CR}$, we obtain from \eqref{Eq_achvbl_err_exp} and \eqref{Eq_err_exp_achv} that
\begin{align*}
P_e \leq
   & \exp\left[- \frac{\ln M}{\CR}\left(\frac{\log e}{2 N_0} - \CR \right) \right], \text{ if } 0 < \CR \leq \frac{1}{4} \frac{\log e}{N_0}
\end{align*}
and
 \begin{align*}
 P_e \leq   & \exp\left\{- \frac{\ln M}{\CR}  \left(\sqrt{\frac{\log e}{N_0}} - \sqrt{ \CR}\right)^2\right\}, \text{ if } \frac{1}{4} \frac{\log e}{N_0}  \leq \CR \leq \frac{\log e}{N_0}.
 \end{align*}  
This proves the upper bounds on the probability of error in~\eqref{Eq_orth_sinlg_uppr1} and \eqref{Eq_orth_sinlg_uppr2}.

\subsection{Lower bounds}
\label{Sec_lowr}
To prove the lower bounds on the probability of error in~\eqref{Eq_orth_sinlg_uppr1} and \eqref{Eq_orth_sinlg_uppr2}, we first argue that, for an orthogonal code, the optimal probability of error is achieved by codewords of equal energy. Then, for any given $\CR$ and an orthogonal codebook where all codewords have equal energy, we derive the lower bound in \eqref{Eq_orth_sinlg_uppr2}, which is optimal at high rates. 
 We further obtain an improved lower bound on the probability of error for low rates. 
 Finally, the lower bound in  \eqref{Eq_orth_sinlg_uppr1} follows by showing that a combination of the two lower bounds yields a lower bound, too.

\subsubsection{Equal-energy codewords are optimal}

We shall argue that, for an orthogonal code with energy upper-bounded by $E_n$, there is no loss in optimality in assuming that all codewords have energy  $E_n$. To this end,
we first note that, without loss of generality, we can restrict ourselves to codewords of the form
\begin{align}
\bx_m = (0,\ldots, \sqrt{E_{\bx_m}},\ldots,0), \; m=1, \ldots, M, \label{Eq_ortho_code}
\end{align}
where  $E_{\bx_m}\leq E_n$ denotes the energy of codeword $\bx_m$. Indeed, any orthogonal codebook can be multiplied by an orthogonal matrix to obtain this form. Since the additive Gaussian noise $\bZ$ is zero mean and has a diagonal covariance matrix, this does not change the probability of error.

To argue that equal energy codewords are optimal, let us consider a code $\cC$ for which some codewords have energy strictly less than  $E_n$. From $\cC$, we can construct a new code $\cC'$ by multiplying each codeword $\bx_m$ by $\sqrt{E_n/E_{\bx_m}}$. Clearly, each codeword in $\cC'$ has energy $E_n$. Let $\bY$ and $\bY'$ denote the channel outputs when we transmit codewords from $\cC$ and $\cC'$, respectively, and let $P_e(\cC)$ and $P_e(\cC')$ denote the corresponding minimum probabilities of error. By multiplying each dimension of the channel output $\bY'$ by $\sqrt{E_{\bx_m}/E_n}$ and adding Gaussian noise of zero mean and variance $E_n/E_{\bx_m}$, we can construct a new channel output $\tilde{\bY}$ that has the same distribution as $\bY$. Consequently, $\cC'$ can achieve the same probability of error as $\cC$ by applying the decoding rule of $\cC$ to $\tilde{\bY}$. It follows that $P_e(\cC')\leq P_e(\cC)$. We conclude that, in order to find lower bounds on the probability of error, it suffices to consider codes whose codewords have energy $E_n$.

\subsubsection{High-rate lower bound}

To obtain the high-rate lower bound \eqref{Eq_orth_sinlg_uppr2}, we follow the analysis given in~\cite{ShannonGB67} (see also~\cite[Section~3.6.1]{ViterbiO79}).  Thus, we shall first derive a lower bound on the maximum probability of error
 \begin{align*}
P_{e_{\max}}  & :=  \max_{m} P_{e_m},
 \end{align*}
 where $P_{e_m}$ denotes the probability of error in decoding message $m$.
 In a second step, we derive from this bound a lower bound on the average probability of error $P_e$ by means of expurgation.
 For $P_{e_{\max}}$, it was shown that at least one of the 
following two inequalities is always satisfied~\cite[Section~3.6.1]{ViterbiO79}:
\begin{align}
1/M & \geq \frac{1}{4} \exp\left[\mu(s) - s\mu'(s) -s \sqrt{2\mu''(s)}\right],  \label{Eq_lowr_first} \\
P_{e_{\max}} & \geq \frac{1}{4} \exp\left[\mu(s) + (1-s) \mu'(s) -(1-s) \sqrt{2\mu''(s)}\right], \label{Eq_lowr_second}
\end{align}
for all $0 \leq s \leq 1$, where
\begin{align}
\mu(s) & =-\frac{E}{N_0} s(1-s), \label{Eq_const1} \\
\mu'(s) & = -\frac{E}{N_0}(1-2s),\label{Eq_const2}\\
\mu''(s) & = \frac{2E}{N_0}. \label{Eq_const3}
\end{align}
By substituting these values in \eqref{Eq_lowr_first}, we obtain
\begin{align*}
\ln M \leq \frac{E}{N_0}\left[s^2 + \frac{2s}{\sqrt{E/N_0}}+ \frac{ \ln 4}{E/N_0}\right].
\end{align*}
Using that $0 \leq s \leq 1$ and that  $E = \frac{\log M}{\CR}$, this can be further upper-bounded by
\begin{align}
\CR & \leq \frac{\log e}{N_0} \left[s^2 + \frac{2}{\sqrt{E/N_0}}+ \frac{ \ln 4}{E/N_0}\right].  \label{Eq_uppr_rate}
\end{align}
Similarly, substituting \eqref{Eq_const1}-\eqref{Eq_const3} in \eqref{Eq_lowr_second} yields 
\begin{align}
P_{e_{\max}}  & \geq \exp\left[- \frac{E}{N_0}(1-s)^2 - 2(1-s)\sqrt{\frac{E}{N_0}} - \ln 4\right] \notag\\
& \geq \exp\left[- \frac{E}{N_0}\left((1-s)^2 + \frac{2}{\sqrt {E/N_0}} +\frac{ \ln 4}{E/N_0}\right)\right]. \label{Eq_low_bnd2_err_prob}
\end{align}
For a given $E$, let $\delta_E$ be defined as  $\delta_E := 2\left(\frac{2}{\sqrt {E/N_0}} +\frac{ \ln 4}{E/N_0}\right)$
 and let $s_{E} := \sqrt{\CR \frac{N_0}{\log e}- \delta_E}$.
Then, \eqref{Eq_uppr_rate} is violated for $s=s_{E}$, which implies that \eqref{Eq_low_bnd2_err_prob} must be satisfied for $s=s_{E}$. By substituting $s=s_E$ in~\eqref{Eq_low_bnd2_err_prob}, we obtain 
\begin{align}
&P_{e_{\max}} 
 \geq \exp\left[- E\left(\left(\sqrt{ \frac{1}{N_0}}- \sqrt{ \frac{\CR}{\log e}  -  \frac{\delta_E}{N_0} }\right)^2 + \frac{\delta_E}{2N_0} \right) \right]. \label{Eq_err_exp_upp}
\end{align}

We next use \eqref{Eq_err_exp_upp} to derive a lower bound on $P_e$. Indeed, any codebook $\cC$ with $M$ messages can be divided into two codebooks $\cC_1$ and $\cC_2$ of $M/2$ messages each. If we divide the codebook such that  $\cC_1$ contains the codewords with the smallest probability of error $P_{e_m}$  and $\cC_2$ contains the codewords with the largest $P_{e_m}$, then it holds for each codeword  in $\cC_1$ that  $ P_{e_m} \leq 2 P_e$. Consequently, the largest error probability of code $\cC_1$, denoted as $P_{e_{\max}}(\cC_1)$,  and the average error probability 
of code $\cC$, denoted as $P_e(\cC)$, satisfy
\begin{align}
P_e(\cC) \geq \frac{1}{2} P_{e_{\max}}(\cC_1). \label{Eq_max_half}
\end{align}
Applying \eqref{Eq_err_exp_upp} for $\cC_1$, and using that the rate per unit-energy of $\cC_1$ satisfies $\CR' = \frac{\log M/2}{E} = \CR - \frac{1}{E}$, we obtain 
\begin{align*}
& P_{e_{\max}} 
  \geq\exp\left[- E\left(\left(\sqrt{ \frac{1}{N_0}}- \sqrt{ \frac{\CR}{\log e} - \frac{1}{E}  -   \frac{\delta_E}{N_0} }\right)^2 + \frac{\delta_E}{2N_0} \right) \right].
\end{align*}
Together with \eqref{Eq_max_half}, this yields
\begin{align}
P_e & \geq \frac{1}{2} \exp\left[- E\left(\left(\sqrt{ \frac{1}{N_0}}- \sqrt{ \frac{\CR}{\log e} - \frac{1}{E}  -   \frac{\delta_E}{N_0} }\right)^2 + \frac{\delta_E}{2N_0} \right) \right] \notag \\
& = \exp\left[- E\left(\left(\sqrt{ \frac{1}{N_0}}- \sqrt{ \frac{\CR}{\log e} - \frac{1}{E}  -   \frac{\delta_E}{N_0} }\right)^2 + \frac{\delta_E}{2N_0} - \frac{\ln 2}{E}  \right) \right].\label{Eq_prob_low_bnd}
\end{align}
Let $\delta'_E := \frac{1}{E}  +   \frac{\delta_E}{N_0} $. Then 
\begin{align*}
\sqrt{ \frac{\CR}{\log e} - \frac{1}{E}  -   \frac{\delta_E}{N_0} } &= \sqrt{ \frac{\CR}{\log e} - \delta'_E}\\
& =  \sqrt{ \frac{\CR}{\log e} } + O(\delta'_{E})\\
&=  \sqrt{ \frac{\CR}{\log e} } + O\left(\frac{1}{\sqrt{E}}\right)
\end{align*}
where the last step follows by noting that $O(\delta'_E)=O(\delta_E)=O(1/\sqrt{E})$. Further defining $ \delta''_E := \frac{\delta_E}{2N_0} - \frac{\ln 2}{E}$, we may write \eqref{Eq_prob_low_bnd} as
\begin{align}
P_e &\geq \exp\left[- E\left(\left(\sqrt{ \frac{1}{N_0}}- \sqrt{ \frac{\CR}{\log e}} + O\left(\frac{1}{\sqrt{E}}\right)\right)^2 + \delta''_E  \right) \right] \notag \\
&  = \exp\left[- E\left(\left(\sqrt{ \frac{1}{N_0}}- \sqrt{ \frac{\CR}{\log e}} \right)^2 +O\left(\frac{1}{\sqrt{E}} \right)\right) \right] \label{Eq_prob_low_bnd1}
\end{align}
since $O(\delta''_E)=O(\delta_E)=O(1/\sqrt{E})$. By substituting $E = \frac{\log M}{\CR}$, \eqref{Eq_prob_low_bnd1} yields
\begin{align}																						
P_{e}  & \geq \exp\left[- \frac{\ln M}{\CR}\left(\left(\sqrt{\frac{\log e}{N_0}}- \sqrt{\CR }\right)^2 + O\left(\frac{1}{\sqrt{E}}\right) \right) \right].\label{Eq_err_exp_lwr}
\end{align}
Since $O\left(\frac{1}{\sqrt{E}}\right) \to 0$ as $E \to \infty$, this  proves the lower bound in \eqref{Eq_orth_sinlg_uppr2}.

\subsubsection{ Low-rate lower bound}
To prove the lower bound in \eqref{Eq_orth_sinlg_uppr1}, we first derive a lower bound on $P_e$ that, for low rates, is tighter than \eqref{Eq_err_exp_lwr}. This bound is based on the fact that for $M$ codewords of energy $E$, the minimum Euclidean distance $d_{\min}$ between  codewords is upper-bounded by $\sqrt{2EM/(M-1)}$~\cite[Section~3.7.1]{ViterbiO79}. Since for the Gaussian channel the maximum error probability is lower-bounded by the largest  pairwise error probability, it follows that 
	\begin{align}
	P_{e_{\max}} & \geq Q\left( \frac{d_{\min}}{\sqrt{2 N_0}}\right) \notag \\
	& \geq Q\left(  \sqrt{ \frac{ EM}{N_0(M-1)}}\right) \notag\\
	& \geq \left(1-\frac{1}{EM/N_0(M-1)}\right) \frac{e^{-\frac{EM}{2N_0(M-1)}}}{\sqrt{2 \pi} \sqrt{ EM/N_0(M-1)}}, \label{Eq_Low_rate_low}
	\end{align}	
where the last inequality follows because~\cite[Section~2.3]{ViterbiO79}
	\begin{align*}
	Q(\beta) \geq \left(1-\frac{1}{\beta^2}\right) \frac{e^{-\beta^2/2}}{\sqrt{2 \pi} \beta}, \quad \beta>0.
	\end{align*}
	Let $\beta_E := \sqrt{EM/N_0(M-1)}$. It follows that
	\begin{align}
		\sqrt{E/N_0} \leq \beta_E \leq \sqrt{2E/N_0}, \quad M\geq 2. \label{Eq_beta_bounds}
	\end{align}
Since by the assumption $\CR = \frac{\log M}{E} >0$ we have that $M \to \infty$ as $E \to \infty$, applying \eqref{Eq_beta_bounds} to \eqref{Eq_Low_rate_low} yields for sufficiently large $E$
	\begin{align}	
	 P_{e_{\max}} & \geq \frac{1}{ \sqrt{2 \pi}} \exp\left[-E\left(\frac{1}{2N_0}\left(1+ \frac{1}{M-1}\right)\right)\right]
		\exp\left[\ln \left(\frac{1}{\beta_E} - \frac{1}{\beta_E^3}\right)\right] \ \notag  \\
	& \geq \frac{1}{ \sqrt{2 \pi}} \exp\left[-E\left(\frac{1}{2N_0}\left(1+ \frac{1}{M-1}\right)\right)\right]
	\exp\left[-E \frac{\frac{3}{2}\ln (2E/N_0)- \ln(E/N_0 -1)}{E}\right] \notag\\
	& =  \exp\left[-E\left(\frac{1}{2N_0}\left(1+ \frac{1}{M-1}\right) + O\left(\frac{\ln E }{E}\right)\right)\right].\label{Eq_lowr_prob}
	\end{align}
Following similar steps of expurgation as before, we obtain from \eqref{Eq_lowr_prob} the lower bound
	\begin{align}
	P_e &\geq \exp\left[-E\left(\frac{1}{2N_0}\left(1+ \frac{1}{\frac{M}{2}-1}\right) + O\left(\frac{\ln E }{E}\right)\right)\right]. \notag 
	\end{align}
By using that $M = 2^{\CR E}$, it follows that	
\begin{align}
P_e &\geq \exp\left[-E\left(\frac{1}{2N_0}\left(1+ \frac{1}{\frac{2^{\CR E}}{2}-1}\right) + O\left(\frac{\ln E }{E}\right)\right)\right]\label{Eq_low_any_rate} 
\end{align}
from which we obtain that, for any rate per unit-energy $\CR > 0$,
\begin{align}
P_e &\geq \exp\left[-E\left(\frac{1}{2N_0} + O\left(\frac{\ln E}{E}\right)\right)\right]. \label{Eq_low_any_rate1} 
\end{align}
\comment{
Next we give a lemma to prove our result.  To prove this, we use the techniques  from \cite{ShannonGB67}.
\begin{lemma}
	\label{Lem_comb_rate}
	For rates $\CR'$ and $\CR''$, assume
	\begin{align*}
	P_e(\CR') &\geq e^{-E\left(E'_{\exp}(\CR')+o(1)\right)},
	\end{align*}
	 and 
	\begin{align*}
	P_e(\CR'') &\geq e^{-E\left(E''_{\exp}(\CR'')+o(1)\right)}.	
	\end{align*}
	Then for orthogonal codebook with rate $\CR$
	$\CR = \lambda\CR'+(1-\lambda)\CR''$, where $0 \leq \lambda \leq 1$, the probability of error is lower bounded by
	$$P_e(\CR) \geq e^{-E[\lambda E'_{\exp}(\CR')+(1-\lambda)E''_{\exp}(\CR'')+o(1)]}.$$
\end{lemma}
}

\subsubsection{Combining the high-rate and the low-rate bounds}
We finally show that a combination of the lower bounds \eqref{Eq_prob_low_bnd1} and \eqref{Eq_low_any_rate} yields a lower bound, too. This then proves the lower bound in \eqref{Eq_orth_sinlg_uppr1}. 

Let $P_e^{\bot}(E,M)$ denote the smallest probability of error that can be achieved by an orthogonal codebook with $M$ codewords of energy $E$. 
We first note that $P_e^{\bot}(E,M)$ is monotonically increasing in $M$. Indeed, without loss of optimality, we can restrict ourselves to codewords of the form \eqref{Eq_ortho_code}, all having energy $E$. In this case, the probability of correctly decoding message $m$ is given by~\cite[Section~8.2]{Gallager68}
\begin{align}
P_{c,m}^{\bot} & = \textnormal{Pr}\biggl(\bigcap_{i\neq m} \{Y_m > Y_i\}\biggm|\bX=\bx_m\biggr) \notag \\
& =\frac{1}{\sqrt{\pi N_0}}\int_{-\infty}^{\infty} \exp\left[ \frac {(y_m - \sqrt{E})^2}{N_0}  \right] \textnormal{Pr}\biggl(\bigcap_{i\neq m}\{Y_i < y_m\}\biggm|\bX=\bx_m\biggr)\d y_m \notag\\
& =\frac{1}{\sqrt{\pi N_0}}\int_{-\infty}^{\infty} \exp\left[ \frac {(y_m - \sqrt{E})^2}{N_0}  \right] \left(1-Q(y_m)\right)^{M-1} \d y_m,
 \label{Eq_prob_ortho_corrct}
\end{align}
where $Y_i$ denotes the $i^\mathrm{th}$ component of the received vector $\bY$. In the last step of \eqref{Eq_prob_ortho_corrct}, we have used that, conditioned on $\bX=\bx_m$, the events $ \{Y_i<y_m\}$, $i\neq m$ are independent and $\textnormal{Pr}(Y_i < y_m|\bX=\bx_m)$ can be computed as $1-Q(y_m)$.
Since $P_{c,m}^{\bot}$ is the same for all $m$, we have $P_e^{\bot}(E,M) = 1 -P_{c,m}^{\bot}$. The claim then follows by observing that \eqref{Eq_prob_ortho_corrct} is monotonically decreasing in $M$.

Let $ \tilde{M}$ be the largest power of 2 less than or equal to $M$. It follows by the monotonicity of $P_e^{\bot}(E,M)$ that
\begin{align}
P_e^{\bot}(E,M) \geq P_e^{\bot}(E, \tilde{M}).\label{Eq_ortho_two_code1}
\end{align} 
We next show that for every $E_1$ and $E_2$ satisfying $E=E_1+E_2$, we have
\begin{align}
P_e^{\bot}(E, \tilde{M}) \geq P_e(E_{1},  \tilde{M}, L)P_e(E_{2}, L+1),\label{Eq_prob_prod1}
\end{align}
where $P_e(E_1,  \tilde{M}, L)$ denotes the smallest probability of error that can be achieved by a codebook with $\tilde{M}$ codewords of energy $E_1$ and a list decoder of list size $L$, and $P_e(E_{2}, L+1)$ denotes the smallest probability of error that can be achieved by a codebook with $L+1$ codewords of energy $E_2$.

To prove \eqref{Eq_prob_prod1}, we follow along the lines of \cite{ShannonGB67}, which showed the corresponding result for codebooks of a given blocklength rather than a given energy. Specifically, it was shown in \cite[Theorem~1]{ShannonGB67} that for every codebook $\cC$ with $M$ codewords of blocklength $n$, and for any $n_1$ and $n_2$ satisfying $n=n_1+n_2$, we can lower-bound the probability of error by
\begin{equation}
P_e(\cC) \geq P_e(n_1,  M,L )P_e(n_2, L+1), \label{Eq_prob_prod2}
\end{equation}
where $P_e(n_1, M, L )$ denotes the smallest probability of error that can be achieved by a codebook with $M$ codewords of blocklength $n_1$ and a list decoder of list size $L$, and $P_e(n_2, L+1)$ denotes the smallest probability of error that can be achieved by a codebook with $L+1$ codewords of blocklength $n_2$.
This result follows by writing the codewords $\bx_m$ of blocklength $n$ as concatenations of the vectors
\begin{align*}
	\bx'_m =(x_{m,1}, x_{m,2}, \ldots, x_{m,n_1})
	\end{align*}
	and 
	\begin{align*}
	\bx''_m =(x_{m,n_1+1}, x_{m,n_1+2}, \ldots,x_{m,n_1+n_2})
	\end{align*}
	and, likewise, by writing the received vector $\by$ as the concatenation of the  vectors $\by'$ and $\by''$ of length $n_1$ and $n_2$, respectively. Defining $\Delta_m$ as the decoding region for message $m$ and $\Delta''_m(\by')$ as the decoding region for message $m$ when $\by'$ was received, we can then write $P_e(\cC)$ as
	\begin{align}
	\label{eq:SGB67}
	P_e(\cC) & = \frac{1}{M} \sum_{m=1}^M \sum_{\by'} p(\by'|\bx'_m) \sum_{\by''\in\bar{\Delta}''_m}p(\by''|\bx''_m)
	\end{align}
where $\bar{\Delta}''_m$ denotes the complement of $\Delta''_m$. Lower-bounding first the inner-most sum in \eqref{eq:SGB67} and then the remaining terms, one can prove \eqref{Eq_prob_prod2}.

A codebook with $\tilde{M}$ codewords of the form \eqref{Eq_ortho_code} can be transmitted in $\tilde{M}$ time instants, since in the remaining time instants all codewords are zero. We can thus assume without loss of optimality that the codebook's blocklength is $\tilde{M}$. Unfortunately, for such codebooks, the above approach yields \eqref{Eq_prob_prod1} only in the trivial cases where either $E_1=0$ or $E_2=0$. Indeed, $E_1$ and $E_2$ correspond to the energies of the vectors $\bx'_m$ and $\bx''_m$, respectively, and for \eqref{Eq_ortho_code} we have $\bx'_m=\mathbf{0}$ if $m>n_1$ and $\bx''_m=\mathbf{0}$ if $m\leq n_1$, where $\mathbf{0}$ denotes the all-zero vector. We sidestep this problem by multiplying the codewords by a normalized Hadamard matrix. The Hadamard matrix, denoted by $H_{j}$, is a square matrix of size $j \times j$ with entries $\pm 1$ and has the property that all rows are orthogonal. Sylvester's construction shows that there exists a Hadamard matrix of order $j$ if $j$ is a power of 2. Recalling that $\tilde{M}$ is a power of $2$, we can thus find a normalized Hadamard matrix \[\tilde{H}:=\frac{1}{\sqrt{\tilde{M}}} H_{\tilde{M}}.\] Since the rows of $\tilde{H}$ are orthonormal, it follows that the matrix $\tilde{H}$ is orthogonal. Further noting that the additive Gaussian noise $\bZ$ is zero mean and has a diagonal covariance matrix, we conclude that the set of codewords $\{\tilde{H}\bx_m,\,m=1,\ldots,\tilde{M}\}$ achieve the same probability of error as the set of codewords $\{\bx_m,\,m=1,\ldots,\tilde{M}\}$. Thus, without loss of generality we can restrict ourselves to codewords of the form $\tilde{\bx}_m=\tilde{H}\bx_m$, where $\bx_m$ is as in \eqref{Eq_ortho_code}. Such codewords have constant modulus, i.e., $|\tilde{x}_{m,k}| = \sqrt{\frac{E}{\tilde{M}}}, k=1, \ldots,\tilde{M}$. This has the advantage that the energies of the vectors
\begin{align*}
	\tilde{\bx}'_m =(\tilde{x}_{m,1}, \tilde{x}_{m,2}, \ldots, \tilde{x}_{m,n_1})
	\end{align*}
	and 
	\begin{align*}
	\tilde{\bx}''_m =(\tilde{x}_{m,n_1+1}, \tilde{x}_{m,n_1+2}, \ldots,\tilde{x}_{m,n_1+n_2}).
	\end{align*}
are proportional to $n_1$ and $n_2$, respectively. Thus, by emulating the proof of \eqref{Eq_prob_prod2}, we can show that for every $n_1$ and $n_2$ satisfying $\tilde{M}=n_1+n_2$ and $E_i=E n_i/\tilde{M}$, $i=1,2$, we have
\begin{equation}
\label{Eq_prob_prod3}
P_e^{\bot}(E, \tilde{M}) \geq P_e(E_{1},n_1,\tilde{M}, L)P_e(E_{2},n_2,L+1),
\end{equation}
where $P_e(E_{1},n_1,\tilde{M}, L)$ denotes the smallest probability of error that can be achieved by a codebook with $\tilde{M}$ codewords of energy $E_1$ and blocklength $n_1$ and a list decoder of list size $L$, and $P_e(E_{2}, n_2, L+1)$ denotes the smallest probability of error that can be achieved by a codebook with $L+1$ codewords of energy $E_2$ and blocklength $n_2$. We then obtain \eqref{Eq_prob_prod1} from \eqref{Eq_prob_prod3} because
\begin{equation*}
P_e(E_{1},n_1,\tilde{M}, L) \geq P_e(E_{1},\tilde{M}, L) \quad \textnormal{and} \quad P_e(E_{2},n_2,L+1) \geq P_e(E_{2},L+1).
\end{equation*}

	We next give a lower bound on $ P_e(E_{1}, \tilde{M}, L)$. Indeed, for list decoding of list size $L$, the inequalities \eqref{Eq_lowr_first} and \eqref{Eq_lowr_second} can be replaced by~\cite[Lemma~3.8.1]{ViterbiO79}
	 \begin{align}
	 L/M  & \geq \frac{1}{4} \exp\left[\mu(s) - s\mu'(s) -s \sqrt{2\mu''(s)}\right],  \label{Eq_lowr_first_list} \\
	 P_{e_{\max}} & \geq \frac{1}{4} \exp\left[\mu(s) + (1-s) \mu'(s) -(1-s) \sqrt{2\mu''(s)}\right]. \label{Eq_lowr_second_list}
	 \end{align}
	 Let $\CR_1 := \frac{ \log (M/L)}{E_{1}}$ and $\tilde{\CR}_1 := \frac{ \log (\tilde{M} /L)}{E_{1}}$. From the definition of $\tilde{M} $, we have $\tilde{M} \leq M \leq 2\tilde{M}$. Consequently,
	 \begin{align}
	 \CR_1 - \frac{1}{E_1} \leq  \tilde{\CR}_1 \leq \CR_1. \label{Eq_rate_low_high}
	 \end{align}
	  By following the steps that led to \eqref{Eq_prob_low_bnd1}, we thus obtain
	 \begin{align}
	 P_e(E_{1}, \tilde{M}, L) &\geq \exp\left[- E_{1}\left(\left(\sqrt{ \frac{1}{N_0}}- \sqrt{ \frac{\tilde{\CR}_1}{\log e}}\right)^2 + O\left(\frac{1}{\sqrt{E_1}}\right) \right) \right] \notag \\
	 & = \exp\left[- E_{1}\left(\left(\sqrt{ \frac{1}{N_0}}- \sqrt{ \frac{\CR_1}{\log e}}\right)^2 + O\left(\frac{1}{\sqrt{E_1}}\right) \right) \right]. \label{Eq_prob_low_bnd2_list}
	 \end{align}
	 
	 To lower-bound $P_e(E_{2}, L+1)$, we apply \eqref{Eq_low_any_rate} with $\CR_2 := \frac{\log  (L+1)}{E_2}$ to obtain
	 \begin{align}
	 P_e(E_{2}, L+1) & \geq \exp\left[-E_2\left(\frac{1}{2N_0}\left(1+ \frac{1}{\frac{2^{\CR_2 E_2}}{2}-1}\right) + O\left(\frac{\ln E_2 }{E_2}\right)\right)\right]. \label{Eq_low_rate_lowr_bnd}
	 \end{align}
	 Let \[\Xi_1(\CR_1):=\left(\sqrt{ \frac{1}{N_0}}- \sqrt{ \frac{\CR_1}{\log e}}\right)^2\] and \[\Xi_2(\CR_2) := \frac{1}{2N_0}\left(1+ \frac{1}{\frac{2^{\CR_2 E_2}}{2}-1}\right).\] Then, by substituting \eqref{Eq_prob_low_bnd2_list} and \eqref{Eq_low_rate_lowr_bnd} in \eqref{Eq_prob_prod1} and using \eqref{Eq_ortho_two_code1}, we get
\begin{align}
P_e^{\bot}(E, M) & \geq  \exp\left[- E_{1}\left(\Xi_1(\CR_1)+ O\left(\frac{1}{\sqrt{E}}\right) \right) \right] \exp\left[- E_{2} \left(\Xi_2(\CR_2)+O\left(\frac{\ln E_2 }{E_2}\right)\right) \right]. \label{Eq_ortho_lowr}
\end{align}

Applying \eqref{Eq_ortho_lowr} with a clever choice of $E_1$ and $E_2$, we can show that the error exponent of $P_e^{\bot}(E, M)$ is upper-bounded by a convex combination of $\Xi_1(\CR_1)$ and $\Xi_2(\CR_2)$. Indeed, let $\lambda := \frac{E_{1}}{E}$. Then, it follows from~\eqref{Eq_ortho_lowr} that
	 \begin{align}
	 P_e^{\bot}(E, M) \geq \exp\left[ -E \left(\lambda \Xi_1(\CR_1)+ (1 - \lambda) \Xi_2(\CR_2)+ O\left(\frac{1}{\sqrt{E}}\right) \right) \right] \label{Eq_prob_prod_low}
	 \end{align}
and
	 \begin{align*}
	 \frac{\log M}{E} & = \frac{\log (M/L) + \log L}{E}\\ 
	 & = \lambda  \frac{\log (M/L)  }{E_{1}} + (1-\lambda)\frac{\log L}{E_{2}} \\
	 & = \lambda \CR_1 + (1-\lambda) \CR_2.
	 \end{align*}

Let $\CR:=\frac{\log M}{E} \leq \frac{\log e}{4N_0}$ and $\gamma_E:=\min\left\{\frac{1}{\sqrt{E}},\frac{\CR}{2}\right\}$. We conclude the proof of the lower bound in \eqref{Eq_orth_sinlg_uppr1} by applying \eqref{Eq_prob_prod_low} with
\begin{align*}
\lambda_E & = \frac{\CR - \gamma_E}{\frac{\log e}{4N_0} - \gamma_E}
\end{align*}
and the rates per unit-energy $\CR_1 = \frac{1}{4} \frac{\log e}{N_0}$ and $\CR_2 =\gamma_E$. It follows that
\begin{align}
P_e^{\bot}(E, M) & \geq \exp\left[-E \left(\frac{\lambda_E}{4N_0}+ \frac{1-\lambda_E}{2N_0}+ \frac{1-\lambda_E}{\frac{2^{\gamma_E E_2}}{2}-1}+O\left(\frac{1}{\sqrt{E}}\right)\right)\right]\notag\\ 
& = \exp\left[-E \left(\frac{1}{2N_0}- \frac{\lambda_E}{4N_0}+ \frac{1-\lambda_E}{\frac{2^{\gamma_E E_2}}{2}-1}+O\left(\frac{1}{\sqrt{E}}\right)\right)\right]. \label{Eq_low_mid_rate}
 \end{align}
Noting that $\lambda_E =  \frac{\CR}{(\log e)/4N_0} + O\left(\frac{1}{\sqrt{E}}\right) $, \eqref{Eq_low_mid_rate} can be written as
\begin{align}
P_e^{\bot}(E, M) \geq \exp \left[-E \left( \frac{1}{2 N_0} - \frac{\CR}{\log e}+ O\left(\frac{1}{\sqrt{E}}\right)\right)\right], \quad 0 < \CR \leq \frac{1}{4} \frac{\log e}{N_0}.
\end{align}
Since $O\left(\frac{1}{\sqrt{E}}\right)\to 0$ as $E\to\infty$, this proves the lower bound in \eqref{Eq_orth_sinlg_uppr1}.

\fi

\bibliography{Bibliography.bib}
\bibliographystyle{IEEEtran}	

\end{document}